\documentclass[runningheads]{llncs}
\usepackage{graphicx}
\usepackage{hyperref}
\usepackage{subcaption}
\usepackage{wrapfig}
\usepackage{amsmath, xspace} % for align and alignat envirionment
\usepackage{dsfont} % mathds for NN
\usepackage{todonotes} 
\usepackage{csquotes}
\usepackage{ltl}
\usepackage{mathtools}
\usepackage{thm-restate}
\usepackage{algorithm}
\usepackage{multicol}
\usepackage{multirow}
\usepackage[noend]{algpseudocode}
\usepackage{tikz}
\usetikzlibrary{hobby, backgrounds, arrows, automata, calc, positioning}
\usepackage{listings}

\algrenewcomment[1]{\hfill{\small \color{blue!70} \(//\) \emph{#1}}}
\algnewcommand{\LineComment}[1]{\Statex {\small \color{blue!70} \(\hspace{14pt} //\) \emph{#1}}}
\algnewcommand{\LineCommentIndent}[1]{\Statex {\small\textsf{} \color{blue!70} \(\hspace{27pt} //\) \emph{#1}}}

\newcommand{\set}[1]{\{#1\}}
\newcommand{\NN}{\mathbb N}
\newcommand{\prefix}{\preceq}
\newcommand{\pow}[1]{2^{#1}}
\newcommand{\cupdot}{\mathbin{\dot\cup}}

\newcommand{\model}{\protect{\ensuremath{\mathcal{M}}}\xspace}

\newcommand{\strat}[2]{(2^{#1})^* \rightarrow 2^{#2}}

\newcommand{\AP}{\mathit{AP}}
\newcommand{\tracevars}{\mathcal{V}}
\newcommand{\traceassign}{\Pi}

\newcommand{\T}{T}

\newcommand{\ldot}{\mathpunct{.}}

\newcommand{\Un}{\LTLuntil}
\newcommand{\X}{\LTLnext}
\newcommand{\Gl}{\LTLglobally}
\newcommand{\F}{\LTLeventually}

\newcommand{\game}{\mathcal G}

\pgfdeclarelayer{background}
\pgfsetlayers{background,main}

\begin{document}
\title{Runtime Enforcement of Hyperproperties\thanks{This work was partially supported by the German Research Foundation (DFG) as part of the Collaborative Research Center ``Foundations of Perspicuous Software Systems'' (TRR 248, 389792660) and by the European Research Council (ERC) Grant OSARES (No. 683300).}}
%\author{First Author\inst{1}\orcidID{0000-1111-2222-3333} \and
%Second Author\inst{2,3}\orcidID{1111-2222-3333-4444} \and
%Third Author\inst{3}\orcidID{2222--3333-4444-5555}}
\author{Norine Coenen\inst{1} \and Bernd Finkbeiner\inst{1} \and Christopher Hahn\inst{1} \and Jana Hofmann\inst{1} \and Yannick Schillo\inst{2}}
\institute{CISPA Helmholtz Center for Information Security, Saarbr\"ucken, Germany,\\
	\email{\{norine.coenen, finkbeiner, christopher.hahn, jana.hofmann\}@cispa.de}
	\and Saarland University, Saarbr\"ucken, Germany,\\
	\email{s8yaschi@stud.uni-saarland.de}}
\authorrunning{Coenen, Finkbeiner, Hahn, Hofmann, Schillo}
%\institute{}
%\email{\{norine.coenen, finkbeiner, christopher.hahn, jana.hofman\}@cispa.saarland, s8yaschi@stud.uni-saarland.de}}
%
\maketitle              % typeset the header of the contribution
\begin{abstract}
An enforcement mechanism monitors a reactive system for undesired behavior at runtime and corrects the system's output in case it violates the given specification.
In this paper, we study the enforcement problem for \emph{hyperproperties}, i.e., properties that relate multiple computation traces to each other. We elaborate the notion of \emph{sound} and \emph{transparent} enforcement mechanisms for hyperproperties in two trace input models:
1) the parallel trace input model, where the number of traces is known a-priori and all traces are produced and processed in parallel and 2) the sequential trace input model, where traces are processed sequentially and no a-priori bound on the number of traces is known.
For both models, we study enforcement algorithms for specifications given as formulas in universally quantified HyperLTL, a temporal logic for hyperproperties.
For the parallel model, we describe an enforcement mechanism based on parity games.
For the sequential model, we show that enforcement is in general undecidable and present algorithms for reasonable simplifications of the problem (partial guarantees or the restriction to safety properties).
Furthermore, we report on experimental results of our prototype implementation for the parallel model.

%\keywords{Hyperproperties \and Runtime Enforcement \and Synthesis \and Parity Games \and HyperLTL.}
\end{abstract}
%
%
%===========================================
\section{Introduction}
\label{sec:Intro}
%===========================================

Runtime enforcement combines the strengths of dynamic and static verification by monitoring the output of a running system and also \emph{correcting} it in case it violates a given specification.
Enforcement mechanisms thus provide formal guarantees for settings in which a system needs to be kept alive while also fulfilling critical properties.
Privacy policies, for example, cannot be ensured by shutting down the system to prevent a leakage: an attacker could gain information just from the fact that the execution stopped.

\begin{figure}[t]
%\begin{wrapfigure}{R}{5.7cm}
	\begin{subfigure}[t]{0.45\textwidth}
		\centering
		\begin{tikzpicture}[->,>=stealth',shorten >=1pt,auto,node distance=2.1cm,semithick,scale=0.9, every node/.style={transform shape}]
		\tikzstyle{every state}=[draw, rectangle, align=center]
		\tikzstyle{every text node part}=[align=center]
		
%		\node         (in)  at (0.3,0.5) {};
%		\node [align=left] (in1)  at (-0,0) {\footnotesize new \\trace?};
		\node [state] (sys) at (2,0) {Reactive \\ System};
		\node [] 	  (dots) at (2,-0.7) {$\vdots$};
		\node [state] (sys2) at (2,-1.6) {Reactive \\ System};
		\node [state, minimum height=5em] (enf) at (4.2,-0.8) {Enforcer}; 
		\node		  (spec) at (4.2, 1.1) {Specification $S$};
		\node [anchor = west] (out) at (6,-0.8) {};
		\node [anchor = west, align=left] (name) at (1.7,-2.6) {System with\\ formal guarantees};
		\node [anchor = west, align=left] (name) at (-0.5,1.1) {Environment};
		
		\path[draw] (0.3,-2) to node[pos=0.2, align=left] {\footnotesize$\mathit{I}_n$} (1.1,-2);
		\path[draw] (0.3,0) to node[pos=0.2] {\footnotesize$\mathit{I}_1$} (1.1,0);
		\path[draw] (sys.east) to node[pos=0.1] {\footnotesize$\mathit{O}_1$} (enf);
		\path[draw] ($(enf.east)+(0,0.5)$) to node[pos=0.6] {\footnotesize $\mathit{O}_1'$} ($(out)+(0,0.5)$);
		\node[right =0.7em of enf, yshift=7pt] (dots2) {$\vdots$};
		\path[draw] ($(enf.east)+(0,-0.5)$) to node[pos=0.6] {\footnotesize $\mathit{O}_n'$} ($(out)+(0,-0.5)$);
		\path[draw] (spec) edge (enf);
		\path[draw] (sys2.east) to node[pos=0.8] {\footnotesize$\mathit{O}_n$} (enf);
		\draw [-,dashed] (1.1,0.7) rectangle (5.1,-3.1);
		
		\end{tikzpicture}
	\caption{Parallel model.}	
	\end{subfigure}
	\hfill
	\begin{subfigure}[t]{0.45\textwidth}
		\centering
		\begin{tikzpicture}[->,>=stealth',shorten >=1pt,auto,node distance=2.1cm,semithick,scale=0.9, every node/.style={transform shape}]
		\tikzstyle{every state}=[draw, rectangle, align=center]
		\tikzstyle{every text node part}=[align=center]
		
		%	\node         (in)  at (0.3,0.5) {};
		%		\node [align=left] (in1)  at (-0,0) {\footnotesize new \\trace?};
		\node [state] (sys) at (2,0) {Reactive \\ System};
		\node [state] (enf) at (4.2,0) {Enforcer}; 
		\node		  (spec) at (4.2, 1.2) {Specification $S$};
		\node [anchor = west] (out) at (5.8,0) {};
		\node [anchor = west, align=left] (name) at (1.7,-1.2) {System with\\ formal guarantees};
		\node [anchor = west, align=left] (name) at (-0.5,1.2) {Environment};
		
		\path[draw] (0.3,-0.5) to node[pos=0.2, align=left] {\footnotesize new \\trace?} (1.1,-0.5);
		\path[draw] (0.3,0.5) to node[pos=0.2] {\footnotesize$\mathit{I}$} (1.1,0.5)
		(sys) edge node[pos=0.5] {\footnotesize$\mathit{O}$} (enf)
		(enf) edge node[pos=0.6] {\footnotesize $\mathit{O}'$} (out);
		\path[draw] (spec) edge (enf);
		
		\draw [-,dashed] (1.1,0.8) rectangle (5.1,-1.7);
		
		\end{tikzpicture}
	\caption{Sequential model.}	
	\end{subfigure}
		
	\caption{Runtime enforcement for a reactive system. In case the input-output-relation would violate the specification $S$, the enforcer corrects the output.}	
	\label{fig:enforcement}
%\end{wrapfigure}
\end{figure}
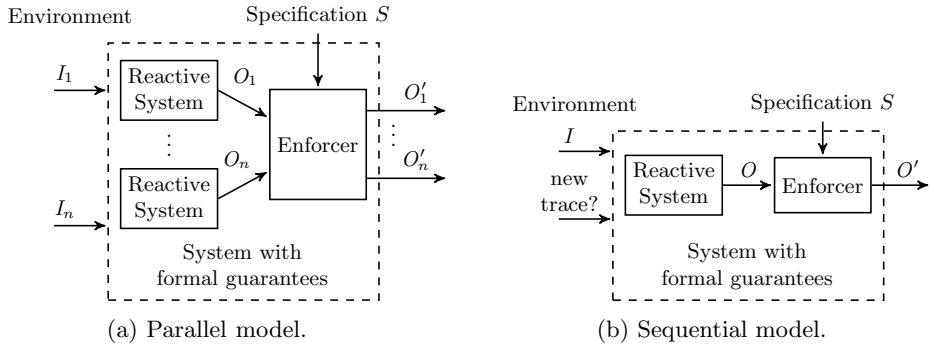

Runtime enforcement has been successfully applied in settings where specifications are given as \emph{trace properties}~\cite{Falcone2011,Tutorial}.
Not every system behavior, however, can be specified as a trace property. 
Many security and privacy policies are \emph{hyperproperties}~\cite{Hyperproperties}, which generalize trace properties by relating multiple execution traces to each other. 
Examples are noninterference~\cite{DBLP:conf/sp/Roscoe95,DBLP:journals/jcs/McLean92}, observational determinism~\cite{DBLP:conf/csfw/ZdancewicM03}, and the detection of out-of-the-ordinary values in multiple data streams~\cite{WearableTracker}.
Previous work on runtime enforcement of hyperproperties either abstractly studied the class of enforceable hyperproperties~\cite{BlackBox} or security policies~\cite{EnforceableSecurityPolicies}, or provided solutions for specific security policies like noninterference~\cite{EnforceableSecurityPolicies,SASI,EditAutomata}.
Our contribution is two-fold.

First, conceptually, we show that hyperproperty enforcement of reactive systems needs to solve challenging variants of the synthesis problem and that the concrete formulation depends on the given trace input model. We distinguish two input models 1) the parallel trace input model, where the number of traces is known a-priori and all traces are produced and processed in parallel and 2)~the sequential trace input model, where traces are processed sequentially and no a-priori bound on the number of traces is known.
Figure~\ref{fig:enforcement} depicts the general setting in these input models.
In the parallel trace input model, the enforcement mechanism observes $n$~traces at the same time. This is, for example, the natural model if a system runs in secure multi-execution~\cite{DBLP:conf/sp/DevrieseP10}.
In the sequential trace input model, system runs are observed in sessions, i.e., one at a time.
An additional input indicates that a new session (i.e., trace) starts.
Instances of this model naturally appear, for example, in web-based applications.
%In Section~\ref{sec:enf-settings}, we describe two possible application scenarios with different trace input models in detail, namely fairness in contract signing and privacy in fitness trackers.

%\vspace{-0.3ex}
Second, algorithmically, we describe enforcement mechanisms for a concrete specification language.
%From a practical point of view, enforcement mechanisms must build on specification languages that can precisely capture application-tailored characteristics.
%When enforcing hyperproperties, it is important, from a practical point of view, to work with a specification language that can precisely capture application-tailored characteristics.
The best-studied temporal logic for hyperproperties is HyperLTL~\cite{DBLP:conf/post/ClarksonFKMRS14}. 
It extends LTL with trace variables and explicit trace quantification to relate multiple computation traces to each other.
HyperLTL can express many standard information flow policies~\cite{DBLP:conf/post/ClarksonFKMRS14}.
In particular, it is flexible enough to state different application-tailored specifications.
%We develop enforcement algorithms for HyperLTL specifications for both trace input models.
We focus on universally quantified formulas, a fragment in which most of the enforceable hyperproperties naturally reside.
For both trace input models, we develop enforcement mechanisms based on parity game solving. For the sequential model, we show that the problem is undecidable in general but provide algorithms for the simpler case that the enforcer only guarantees a correct continuation for the rest of the current session. Furthermore, we describe an algorithm for the case that the specification describes a safety property.
Our algorithms monitor for \emph{losing} prefixes, i.e., so-far observed traces for which the system has no winning strategy against an adversarial environment.
We ensure that our enforcement mechansims are \emph{sound} by detecting losing prefixes at the earliest possible point in time. Furthermore, they are \emph{transparent}, i.e., non-losing prefixes are not altered.

We accompany our findings with a prototype implementation for the parallel model and conduct two experiments: 1) we enforce symmetry in mutual exclusion algorithms and 2) we enforce the information flow policy observational determinism. We will see that enforcing such complex HyperLTL specifications can scale to large traces once the initial parity game solving succeeds. %When no parity game solving is needed, i.e., if the system is non-reactive, our prototype will scale to a larger number of traces with a trade-off of a computational overhead once the traces get longer. 

\paragraph{Related Work.}
\label{sec:related}
HyperLTL has been studied extensively, for example, its expressiveness~\cite{Bozzelli,Hierarchy} as well as its 
%satisfiability~\cite{deciding_hyperproperties,EAHyper}, 
verification~\cite{DBLP:conf/post/ClarksonFKMRS14,MCHyper,VerifyingHyperliveness,QuantitativeModelChecking}, 
synthesis~\cite{DBLP:conf/cav/FinkbeinerHLST18}, 
and monitoring problem~\cite{DBLP:conf/fm/StuckiSSB19,monitoring_hyperproperties_journal,RVHyper,constraint_based_monitoring,DBLP:conf/csfw/BonakdarpourF18}. 
% \cite{DBLP:conf/fm/StuckiSSB19,monitoring_hyperproperties_journal,RVHyper,constraint_based_monitoring,DBLP:conf/csfw/BonakdarpourF18,DBLP:conf/tacas/BrettSB17}.
Especially relevant is the work on realizability monitoring for LTL~\cite{monitoring_realizability} using parity games. 
%We generalize the approach to HyperLTL and extend it to solve the enforcement problem. 
%Similarly, safety games~\cite{ShieldSynthesis,EnforcingUnderBurstError} 
%and B\"uchi games~\cite{BuechiEnforcement} have been used for reactive runtime enforcement. 
Existing work on runtime enforcement includes algorithms for safety properties~\cite{ShieldSynthesis,EnforcingUnderBurstError}, 
real-time properties~\cite{TimedProperties,OnRETimedProperties}, 
%~\cite{TimedProperties,CPSTimedProperties,OnRETimedProperties,RETimedProperties}, 
concurrent software specifications~\cite{EnforceMOP},
%~\cite{EnforceMOP,ConcurrencyEnforcement}, 
and concrete security policies~\cite{EnforceableSecurityPolicies,SASI,EditAutomata}.
For a tutorial on variants of runtime enforcement see~\cite{Tutorial}. 
Close related work is~\cite{BlackBox}, which also studies the enforcement of general hyperproperties but independently of a concrete specification language (in contrast to our work).
Systems are also assumed to be reactive and black-box, but there is no distinction between different trace input models. 
While we employ parity game solving, their enforcement mechanism executes several copies of the system to obtain executions that are related by the specification.

%===========================================
\section{Preliminaries}
\label{sec:prelims}
%===========================================
Let $\Sigma$ be an alphabet of atomic propositions. We assume that $\Sigma$ can be partitioned into inputs and outputs, i.e., $\Sigma = I \cupdot O$. A finite sequence $t \in (\pow{\Sigma})^*$ is a \emph{finite trace}, an infinite sequence $t \in (\pow{\Sigma})^\omega$ is an \emph{infinite trace}. 
We write $t[i]$ for the $(i+1)$-th position of a trace, $t[0..i]$ for its prefix of length $i+1$, and $t[i, \infty]$ for the suffix from position~$i$.
%For the concatenation of two traces $t_1$ and $t_2$ (where $t_1$ is required to be finite) we use $t_1 \cdot t_2$.
A hyperproperty $H$ is a set of sets of infinite traces. 

%--------------------------------------------------------------------------
\paragraph{HyperLTL.} 
%--------------------------------------------------------------------------
HyperLTL~\cite{DBLP:conf/post/ClarksonFKMRS14} is a linear temporal hyperlogic that extends LTL~\cite{LTL} with prenex trace quantification. The syntax of HyperLTL is given with respect to an alphabet $\Sigma$ and a set $\tracevars$ of trace variables.
\begin{align*}
\varphi &\Coloneqq \forall \pi \ldot \varphi \mid \exists \pi \ldot \varphi \mid \psi \\
\psi &\Coloneqq a_\pi \mid \neg \psi \mid \psi \lor \psi \mid \X \psi \mid \psi \Un \psi 
\end{align*}
where $a \in \Sigma$ and $\pi \in \tracevars$. The atomic proposition $a$ is indexed with the trace variable $\pi$ it refers to. We assume that formulas contain no free trace variables. 
HyperLTL formulas are evaluated on a set $T \subseteq (\pow{\Sigma})^\omega$ of infinite traces and a trace assignment function $ \traceassign: \tracevars \rightarrow T$. 
We use $\traceassign[\pi \mapsto t]$ for the assignment that returns $\traceassign(\pi')$ for $\pi' \neq \pi$ and $t$ otherwise. 
Furthermore, let $\Pi[i, \infty]$ be defined as $\Pi[i, \infty](\pi) = \Pi(\pi)[i, \infty]$.
The semantics of HyperLTL is defined as follows: 
\begin{alignat*}{3}
& \T, \traceassign \models a_\pi &&\text{ iff } && a \in \traceassign(\pi)[0] \\
& \T, \traceassign \models \neg \psi &&\text{ iff } && \T, \traceassign \not\models \psi \\
& \T, \traceassign \models \psi_1 \lor \psi_2 &&\text{ iff } && \T, \traceassign \models \psi_1 \text{ or } \T, \traceassign \models \psi_2 \\
& \T, \traceassign \models \X \psi &&\text{ iff } && \T, \traceassign[1, \infty] \models \psi \\
& \T, \traceassign \models \psi_1 \Un \psi_2 &&\text{ iff } && \exists i \geq 0 \ldot \T, \traceassign[i, \infty] \models \psi_2 \text{ and } \forall 0 \leq j < i \ldot \T, \traceassign[j, \infty] \models \psi_1 \\
& \T, \traceassign \models \exists \pi \ldot \varphi &&\text{ iff } && \exists t \in \T \ldot \T, \traceassign[\pi \mapsto t] \models \varphi \\
& \T, \traceassign \models \forall \pi \ldot \varphi &&\text{ iff } && \forall t \in \T \ldot \T, \traceassign[\pi \mapsto t] \models \varphi 
\end{alignat*}
We also use the derived boolean connectives $\land, \rightarrow, \leftrightarrow$ as well as the derived temporal operators $\F \varphi \equiv \mathit{true} \Un \varphi$, $\Gl \varphi \equiv \neg (\F \neg \varphi)$, and $\varphi \LTLweakuntil \psi \equiv (\varphi \Un \psi) \lor \Gl \varphi$. 
A trace set $T$ satisfies a HyperLTL formula $\varphi$ if $T, \emptyset \models \varphi$, where $\emptyset$ denotes the empty trace assignment.

\paragraph{Parity Games.} A parity game $\game$ is a two-player game on a directed graph arena, where the states $V = V_0 \cupdot V_1$ are partitioned among the two players $P_0$ and $P_1$.
States belonging to $P_0$ and $P_1$ are required to alternate along every path.
States are labeled with a coloring function $c : V \rightarrow \NN$.
Player $P_0$ wins the game if they have a strategy to enforce that the highest color occurring infinitely often in a run starting in the initial state is even.
The winning region of a parity game is the set of states from which player $P_0$ has a winning strategy.
%The exact complexity of parity game solving, i.e., determining the winning region, is still unknown. Current state-of-the-art algorithms perform in quasi-polynomial time~\cite{DBLP:conf/mfcs/Parys19}.
A given LTL formula $\varphi$ can be translated to a parity game $\game_\varphi$ in doubly-exponential time~\cite{DBLP:conf/tacas/EsparzaKRS17}. 
Formula~$\varphi$ is realizable iff player $P_0$ wins the game $\game_\varphi$. 
Its winning strategy $\sigma_0$ induces the reactive strategy $\sigma$ representing a system implementation that satisfies $\varphi$.
%A detailed definition of parity games can be found in Appendix~\ref{app:parity}.
%For details on parity game solving and the connection with reactive synthesis we refer to~\cite{mcnaughton1965finite,DBLP:series/natosec/Finkbeiner16,gradel2002automata}.

%===========================================
\section{Hyperproperty Enforcement}
\label{sec:hyper_enforcement}
%===========================================
%An enforcement mechanism for hyperproperties has to monitor \emph{multiple} traces of a system and compare them against each other in order to intervene when the system would violate the specification. Enforcement thus combines multiple formal problems, which makes it especially hard, conceptually as well as algorithmically.
In this section, we develop a formal definition of hyperproperty enforcement mechanisms for reactive systems modeled with the parallel and the sequential trace input model.
To this end, we first formally describe reactive systems under the two trace input models by the \emph{prefixes} they can produce. 
Next, we develop the two basic requirements on enforcement mechanisms, \emph{soundness} and \emph{transparency} \cite{Tutorial,DBLP:journals/sttt/FalconeFM12}, for our settings.
Soundness is traditionally formulated as
%\todo{J: Das ist kein 1:1 Zitat aus den Papieren. Beste Form hierfür?}
\begin{displayquote}
	the enforced system should be correct w.r.t. the specification.
\end{displayquote}
Transparency (also known as precision~\cite{BlackBox}) states that 
\begin{displayquote}
	the behavior of the system is modified in a minimal way, i.e., the longest correct prefix should be preserved by the enforcement mechanism.
\end{displayquote}
%Soundness means that the enforcer produces correct systems, whereas transparency (sometimes also called precision~\cite{BlackBox}) formulates that the enforcer should not interrupt a correct system.
%In this section, we develop in a step-by-step manner the definition of sound and transparent enforcement mechanisms for hyperproperties in reactive systems modeled with two trace input models.
In the context of reactive systems, formal definitions for soundness and transparency need to be formulated in terms of strategies that describe how the enforcement mechanism reacts to the inputs from the environment and outputs produced by the system.
We therefore define soundness and transparency based on the notion of \emph{losing prefixes} (i.e. prefixes for that no winning strategy exists) inspired by work on monitoring reactive systems~\cite{monitoring_realizability}. 
We will see that the definition of losing prefixes depends heavily on the chosen trace input model.
Especially the sequential model defines an interesting new kind of synthesis problem, which varies significantly from the known HyperLTL synthesis problem.
%Lastly, we show that finding a sound and transparent enforcer is exactly the problem of finding a winning strategy in the respective trace input model.
%We base our definitions on reactive systems.

As is common in the study of runtime techniques for reactive systems, we make the following reasonable assumptions.
First, reactive systems are treated as \emph{black boxes}, i.e., two reactive systems with the same observable input-output behavior are considered to be equal. Thus, enforcement mechanisms cannot base their decisions on implementation details.
Second, w.l.o.g. and to simplify presentation, we assume execution traces to have \emph{infinite length}. 
Finite traces can always be interpreted as infinite traces, e.g., by adding $\mathit{end}^\omega$. 
To reason about finite traces, on the other hand, definitions like the semantics of HyperLTL would need to accommodate many special cases like traces of different lengths.
Lastly, we assume that control stays with the enforcer after a violation occurred instead of only correcting the error and handing control back to the system afterwards. Since we aim to provide formal guarantees, these two problems are equivalent: if only the error was corrected, the enforcement mechanism would still need to ensure that the correction does not make the specification unrealizable in the future, i.e., it would need to provide a strategy how to react to all future inputs.

\subsection{Trace Input Models}
We distinguish two \emph{trace input models}~\cite{monitoring_hyperproperties_journal}, the parallel and the sequential model.
The trace input models describe how a reactive system is employed and how its traces are obtained (see Figure~\ref{fig:enforcement}). 
We formally define the input models by the prefixes they can produce. The definitions are visualized in Figure~\ref{fig:prefixes}.
In the parallel model, a fixed number of $n$ systems are executed in parallel, producing $n$ events at a time.
%Prefixes in the parallel model are defined as follows.
\begin{definition}[Prefix in the Parallel Model]
	An $n$-tuple of finite traces $ U = (u_1, \ldots u_n) \in ((\pow{\Sigma})^*)^n $ is a prefix of $V = (v_1, \ldots v_n) \in ((\pow{\Sigma})^{*/\omega})^n $ (written $U \prefix V$) in the parallel model with $n$ traces iff each $u_i$ is a prefix of $v_i$ (also denoted by $u_i \prefix v_i$). 
\end{definition} 
\begin{figure}[t]
	%\begin{wrapfigure}{R}{5.7cm}
	\begin{subfigure}[t]{0.48\textwidth}
		\centering
		\begin{tikzpicture}[-,shorten >=1pt,auto,line width=1.7pt]
			\draw[opacity=0.3] (0,0) to (2.5,0);
			\draw[opacity=0.3, densely dashed] (2.5,0) to (5,0);
			\node[opacity=0.3] at (5.3,0) {$\mathbf{\dots}$};
					
			\draw[opacity=0.3] (0,0.2) to (2.5,0.2);
			\draw[opacity=0.3, densely dashed] (2.5,0.2) to (5,0.2);
			\node[opacity=0.3] at (5.3,0.2) {$\mathbf{\dots}$};
			
			\draw[opacity=0.3] (0,0.4) to (2.5,0.4);
			\draw[opacity=0.3, densely dashed] (2.5,0.4) to (5,0.4);
			\node[opacity=0.3] at (5.3,0.4) {$\mathbf{\dots}$};
			
			\draw[opacity=0.3] (0,0.6) to (2.5,0.6);
			\draw[opacity=0.3, densely dashed] (2.5,0.6) to (5,0.6);
			\node[opacity=0.3] at (5.3,0.6) {$\mathbf{\dots}$};
			
			\node (U) at (1.25,0.3) {\large$\mathbf{U}$};
			\node (V) at (3.75,0.3) {\large$\mathbf{V}$};
			\node (Vphantom) at (3.5,0.1) {\phantom{\large$\mathbf{V}$}};
		\end{tikzpicture}
		\caption{Prefix in the parallel model.}	
	\end{subfigure}
	\hfill
	\begin{subfigure}[t]{0.48\textwidth}
		\centering
		\begin{tikzpicture}[-,shorten >=1pt,auto,line width=1.7pt]
			\draw[opacity=0.3, densely dashed] (0,0) to (5,0);
			\node[opacity=0.3] at (5.3,0) {$\mathbf{\dots}$};
			
			\draw[opacity=0.3] (0,0.2) to (2.5,0.2);
			\draw[opacity=0.3, densely dashed] (2.5,0.2) to (5,0.2);
			\node[opacity=0.3] at (5.3,0.2) {$\mathbf{\dots}$};
			
			\draw[opacity=0.3] (0,0.4) to (5,0.4);
			\node[opacity=0.3] at (5.3,0.4) {$\mathbf{\dots}$};			
			\draw[opacity=0.3] (0,0.6) to (5,0.6);
			\node[opacity=0.3] at (5.3,0.6) {$\mathbf{\dots}$};
			
			\node (U) at (1,0.4) {\large$\mathbf{U}$};
			\node (V) at (3.5,0.1) {\large$\mathbf{V}$};
		\end{tikzpicture}
		\caption{Prefix in the sequential model.}	
	\end{subfigure}
	
	\caption{Visualization of prefixes in trace input models.}	
	\label{fig:prefixes}
\end{figure}
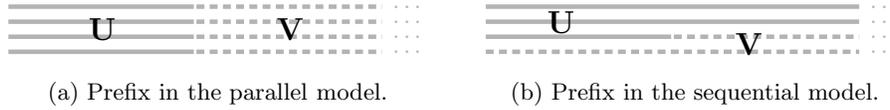
The prefix definition models the allowed executions of a system under the parallel trace input model: 
If the system produces $U$ and after a few more steps produces $V$, then $U \prefix V$. 
Note that the prefix definition is transitive: $U$ can be a prefix of another prefix (then the traces in $V$ are of finite length) or a prefix of infinite-length traces.

In the sequential model, the traces are produced one by one and there is no a-priori known bound on the number of traces.
%A prefix in the sequential model is, thus, defined as follows.
\begin{definition}[Prefix in the Sequential Model]
	Let $U = (u_1, \ldots , u_n) \in ((\pow{\Sigma})^\omega)^*$ be a sequence of traces and $u \in (\pow{\Sigma})^*$ be a finite trace. Let furthermore $V = (v_1, \ldots, v_n, \ldots)$ be a (possibly infinite) sequence of traces with $v_i \in (\pow{\Sigma})^\omega$, and $v \in (\pow{\Sigma})^*$ be a finite trace.	
	We call $(U,u)$ a prefix of $(V,v)$ (written $(U,u) \prefix (V,v)$) iff
	either 1) $U = V$ and $u \prefix v$ or 2)~$V = u_1, \ldots, u_n, v_{n+1}, \ldots$ and $u \prefix v_{n+1}$.	
\end{definition}
We additionally say that $(U, u) \prefix V$ if $(U,u) \prefix (V, \epsilon)$, where $\epsilon$ is the empty trace. 
%As for the parallel trace input model, prefixes describe the allowed executions of a system.
To continue an existing prefix $(U,u)$, the system either extends the started trace $u$ or finishes $u$ and continues with additional traces. Traces in $U$ are of infinite length and describe finished sessions. This means that they cannot be modified after the start of a new session.
Again, prefixes in this model are transitive and are also defined for infinite sets.

\remark
We defined prefixes tailored to the trace input models to precisely capture the influence of the models on the enforcement problem.
Usually, a set of traces $T$ is defined as prefix of a set of traces $T'$ if and only if $\forall t \in T.~ \exists t' \in T'.~ t \prefix t'$~\cite{Hyperproperties}.
%With our definitions, we are able to precisely capture the influence of the different models on the enforcement problem (which is, as we will seen, significant).
A prefix in the sequential model, however, \emph{cannot} be captured by the traditional prefix definition, as it does not admit infinite traces in a prefix.

\subsection{Losing Prefixes for Hyperproperties}
\label{subsec:losing_prefixes}

\label{sec:InputModels}
Losing prefixes describe \emph{when} an enforcer has to intervene based on possible strategies for future inputs.
%We define the notion of \emph{losing} prefixes for hyperproperties for both input models.
As we will see, the definition of losing prefixes, and thus the definition of the enforcement problem, differs significantly for both input models.
For the rest of this section, let $H$ denote an arbitrary hyperproperty.

%We start with the parallel model. Let $n$ be the number of parallel sessions.
We first define strategies for the parallel model with $n$ parallel sessions.
In the enforcement setting, a strategy receives a previously recorded prefix. 
Depending on that prefix, the enforcer's strategy might react differently to future inputs.
We therefore define a \emph{prefixed strategy} as a higher-order function $\sigma:((\pow{\Sigma})^*)^n \rightarrow ((\pow{I})^*)^n \rightarrow (\pow{O})^n$ over $\Sigma = I \cupdot O$. 
The strategy first receives a prefix (produced by the system), then a sequence of inputs on all $n$ traces, and reacts with an output for all traces.
We define a losing prefix as follows.

%\begin{definition}[Strategy in the Parallel Model]
%	A strategy $\sigma$ over $\Sigma = I \cupdot O$ in the parallel model with $n$ sessions is a function $((\pow{\Sigma})^*)^n \rightarrow ((\pow{I})^*)^n \rightarrow (\pow{O})^n$.
%\end{definition}

%
\begin{definition}[Losing Prefix in the Parallel Model]
	A strategy $\sigma(U)$ is losing for $H$ with $U = (u_1, \ldots , u_n) \in ((\pow{\Sigma})^*)^n$ if there are input sequences $(v_1, \ldots , v_n) \in ((\pow{I})^\omega)^n$ such that the following set is not in $H$:
	\begin{align*}
	\bigcup_{1 \leq i \leq n} \set{u_i \cdot (v_i[0] \cup \sigma_U(\epsilon)(i)) \cdot (v_i[1] \cup \sigma_U(v_i[0])(i)) \cdot (v_i[2] \cup \sigma_U(v_i[0]v_i[1])(i)) \ldots}, 
	\end{align*}
	where $\sigma_U = \sigma(U)$ and $\sigma_U(\cdot)(i)$ denotes the $i$-th output that $\sigma$ produces.
	\end{definition}
	We say that $\sigma(U)$ is winning if it is not losing. A prefix $U$ is winning if there is a strategy $\sigma$ such that $\sigma(U)$ is winning. Lastly, $\sigma$ is winning if $\sigma(\epsilon)$ is winning and for all non-empty winning prefixes $U$, $\sigma(U)$ is winning.
%\end{definition}

Similar to the parallel model, a prefixed strategy in the sequential model is a function $\sigma : ((\pow{\Sigma})^\omega)^* \times (\pow{\Sigma})^* \rightarrow (\pow{I})^* \rightarrow \pow{O}$ over $\Sigma = I \cupdot O$.
%Note that this strategy has to cope with an unbounded amount of information about prior executions of the system.
The definition of a losing prefix is the following.
\begin{definition}[Losing Prefix in the Sequential Model]
	\label{def:losing_prefix_unbounded_sequential_model}
	In the sequential model, a strategy $\sigma$ is losing with a prefix $(U, u)$ for $H$, if there are input sequences $V = (v_0, v_1, \ldots)$ with $v_i \in (\pow{I})^\omega$, such that the set $U \cup \set{t_0, t_1, \ldots}$ is not in $H$, where $t_0, t_1, \ldots$ are defined as follows.
	\begin{alignat*}{2}
	&t_0 \coloneqq ~ &&u \cdot (v_0[0] \cup \sigma(U, u)(\epsilon)) \cdot (v_0[1] \cup \sigma(U, u)(v_0[0])) \cdot \ldots \\
	& t_1 \coloneqq &&(v_1[0] \cup \sigma(U \cup \{t_0\}, \epsilon)(\epsilon)) \cdot (v_1[1] \cup \sigma(U \cup \{t_0\}, \epsilon)(v_1[0])) \cdot \ldots \\
	& t_2 \coloneqq &&(v_2[0] \cup \sigma(U \cup \{t_0, t_1\}, \epsilon)(\epsilon)) \cdot  (v_2[1] \cup \sigma(U \cup \{t_0, t_1\}, \epsilon)(v_2[0])) \cdot \ldots
	\end{alignat*}
	\end{definition}
	Winning prefixes and strategies are defined analogously to the parallel model.

\begin{remark}
	The above definitions illustrate that enforcing hyperproperties in the sequential model defines an intriguing but complex problem. Strategies react to inputs based on the observed prefix. The \emph{same} input sequence can therefore be answered differently in the first session and, say, in the third session. The enforcement problem thus not simply combines monitoring and synthesis but formulates a different kind of problem.
\end{remark}

\subsection{Enforcement Mechanisms}
With the definitions of the previous sections, we adapt the notions of sound and transparent enforcement mechanisms to hyperproperties under the two trace input models. 
We define an enforcement mechanism $\mathit{enf}$ for a hyperproperty $H$ to be a computable function which transforms a black-box reactive system $S$ with trace input model~\model into a reactive system $\mathit{enf}(S)$ with the same input~model.
%The following definitions capture soundness and transparency for our setting.
%When applying an enforcer for $H$ to a reactive system, the resulting reactive system should always produce traces that comply with the specification, i.e., it should be sound.
\begin{definition}[Soundness]
	$\mathit{enf}$ is \emph{sound} if for all reactive systems $S$ and all input sequences in model \model, the set of traces produced by $\mathit{enf}(S)$ is in $H$.
\end{definition}

%Further, we require enforcement mechanisms to be transparent; meaning that the enforcer is only allowed to intervene at the latest possible point in time. As long as the original system might still comply to the hyperproperty, the outputs are note altered by the enforcer.
%Formally, transparency is defined as follows.
\begin{definition}[Transparency]
%	An enforcement mechanism $\mathit{enf}$ for a trace input model $\mathcal{M}$ is \emph{transparent} if $\forall S\ldot \forall ~ V,~ V' \in \mathcal{P}^*(I)$, if $V + S(V) \prefix_\mathcal{M} V' + S(V')$ and $V + S(V)$ is not losing in $\mathcal{M}$, then $V+\mathit{enf}(S)(V) \prefix_\mathcal{M} V'+\mathit{enf}(S)(V')$.
%	\todo{provisorische Notation. V+S(V) heißt Input Prefix V zsm mit den Outputs die das System S produziert}
	$\mathit{enf}$ is \emph{transparent} if the following holds:
	Let $U$ be a prefix producible by $S$ with input sequence $s_I$.
	If $U$ is winning, then for any prefix $V$ producible by $\mathit{enf}(S)$ with input sequence $s_I'$ where $s_I \prefix s_I'$, it holds that $U \prefix V$.
\end{definition}

We now have everything in place to define when a hyperproperty is enforcable for a given input model.
\begin{definition}[Enforceable Hyperproperties]
	A hyperproperty $H$ is \emph{enforceable} if there is a sound and transparent enforcement mechanism.
\end{definition}

It is now straightforward to see that in order to obtain a sound and transparent enforcement mechanism, we need to construct a winning strategy for $H$.

\begin{restatable}{proposition}{enforcestrategy}
	\label{prop:enforce_strategy}
	Let $H$ be a hyperproperty and \model be an input model. Assume that it is decidable whether a prefix $U$ is losing in model \model for $H$. Then there exists a sound and transparent enforcement mechanism $\mathit{enf}$ for $H$ iff there exists a winning strategy in \model for $H$.
\end{restatable}
The above proposition describes how to construct enforcement algorithms: 
We need to solve the synthesis problem posed by the respective trace input model.
However, we have to restrict ourselves to properties that can be monitored for losing prefixes.
This is only natural: for example, the property expressed by the HyperLTL formula $\exists \pi \ldot \LTLglobally a_\pi$ can in general not be enforced since it contains a hyperliveness~\cite{Hyperproperties} aspect: 
There is always the possibility for the required trace $\pi$ to occur in a future session (c.f. monitorable hyperproperties in~\cite{DBLP:conf/csfw/AgrawalB16,monitoring_hyperproperties_journal}).
We therefore describe algorithms for HyperLTL specifications from the universal fragment $\forall \pi_1 \ldot \ldots \forall \pi_k \ldot \varphi$ of HyperLTL. Additionally, we assume that the specification describes a property whose counterexamples have losing prefixes.

Before jumping to concrete algorithms, we describe two example scenarios of hyperproperty enforcement with different trace input models.

\begin{example}[Fairness in Contract Signing]
	Contract signing protocols let multiple parties negotiate a contract. 
	In this setting, fairness requires that in every situation where Bob can obtain Alice's signature, Alice  must also be able to obtain Bob's signature. 
	Due to the asymmetric nature of contract signing protocols (one party has to commit first), fairness is difficult to achieve~(see, e.g., \cite{DBLP:journals/jcs/NormanS06}). 
	Many protocols rely on a trusted third party (TTP) to guarantee fairness. 
	The TTP may negotiate
	multiple contracts in parallel sessions. 
	The natural trace input model is therefore the parallel model.
%	The TTP communicates with the involved agents by receiving inputs and then responding with signed contracts where appropriate.
%	In each session, agents either take the role of Alice or Bob. 
%
	Fairness forbids the existence of two traces $\pi$ and $\pi'$ that have the same prefix of inputs, followed in $\pi$ by Bob requesting ($R^\mathit{B}$) and receiving the signed contract ($S^B$), and in $\pi'$ by Alice requesting ($R^\mathit{A}$), but \emph{not} receiving the signed contract ($\neg S^A$): 
%	Fairness forbids the existence of such a pair of traces:
	\[
	\forall \pi \ldot \forall \pi' \ldot \neg ((\bigwedge_{i \in I} (i_\pi \leftrightarrow i_{\pi'})) ~\LTLuntil~ (R^\mathit{B}_\pi \wedge  R_{\pi'}^\mathit{A} \wedge \X (S_\pi^\mathit{B} \wedge \neg S_{\pi'}^\mathit{A})))
	\]

\end{example}

\begin{example}[Privacy in Fitness Trackers]	
	Wearables track a wide range of extremely private health data which can leak an astonishing amount of insight into your health. 
	For instance, it has been found that observing out-of-the-ordinary heart rate values correlates with diseases like the common cold or even Lyme disease~\cite{WearableTracker}. 
%	While this data is very valuable in the hands of your GP, it should never leak to untrusted parties.	
	Consider the following setting. 
	A fitness tracker continuously collects data that is stored locally on the user's device. 
	Additionally, the data is synced with an external cloud.
	While locally stored data should be left untouched, uploaded data has to be enforced to comply with information flow policies.
	Each day, a new stream of data is uploaded, hence the sequential trace input model would be appropriate.
	Comparing newer streams with older streams allows for the detection of anomalies. 
%	Any party with access to the data may be able to infer whether the user is ill. 
%	An enforcer monitoring the data when they are uploaded to the server could change some of the outgoing data to preserve privacy.
%	The comparison of data streams is a hyperproperty. 
%	The base heart rate depends on age, gender, and fitness level of the user.
%	Out-of-the-ordinary values can only be identified by comparison with data from the weeks before. 
	We formalize an exemplary property of this scenario in HyperLTL. 
	Let $\mathit{HR}$ be the set of possible heart rates. Let furthermore $\mathit{active}$ denote whether the user is currently exercising. Then the following property ensures that unusually high heart rate values are not reported to the cloud:
	\[
	\forall \pi \ldot \forall \pi' \ldot \Gl ((\mathit{active}_\pi \leftrightarrow \mathit{active}_{\pi'}) \rightarrow \bigwedge_{r \in \mathit{HR}} (r_\pi \leftrightarrow r_{\pi'})) 
	\]

\end{example}

\section{Enforcement Algorithms for HyperLTL Specifications}
\label{sec:algorithms}

For both trace input models, we present sound and transparent enforcement algorithms for universal HyperLTL formulas defining hyperproperties with losing prefixes. 
First, we construct an algorithm for the parallel input model based on parity game solving.
%The game can be created and solved before the system is executed. 
%The remaining overhead at runtime is very small (see Section~\ref{sec:experiments}).
%For the special case where the system is non-reactive, i.e., where inputs and outputs are not distinguished, we sketch an alternative algorithm.
%
For the sequential trace input model, we first show that the problem is undecidable in the general case.
%We proceed to show that enforcement in the sequential setting is decidable if we assume a bounded number of sessions, although with non-elementary complexity. This is due to a strong correspondence with the synthesis problem of hierarchical distributed architectures.
Next, we provide an algorithm that only finishes the remainder of the current session. 
This simplifies the problem because the existence of a correct future session is not guaranteed.
For this setting, we then present a simpler algorithm that is restricted to safety specifications.
%As for the parallel model, we also consider non-reactive systems.

%===========================================
%\subsection{Algorithms for the Parallel Trace Input Model}
\subsection{Parallel Trace Input Model}
\label{sec:parallel-algorithms}
%===========================================
In short, we proceed as follows: 
First, since we know the number of traces, we can translate the HyperLTL formula to an equivalent LTL formula. 
For that formula, we construct a realizability monitor based on the LTL monitor described in~\cite{monitoring_realizability}.
The monitor is a parity game, which we use to detect minimal losing prefixes and to provide a valid continuation for the original HyperLTL formula.

Assume that the input model contains $n$ traces.
Let a HyperLTL formula $\forall \pi_1 \ldots \forall \pi_k \ldot \varphi$ over $\Sigma = I \cupdot O$ be given, where $\varphi$ is quantifier free. 
We construct an LTL formula $\varphi^n_\text{LTL}$ over $\Sigma' = \set{a_i ~|~ a \in \Sigma, 1 \leq i \leq n}$ as follows:
\begin{align*}
\varphi^n_\text{LTL} \coloneqq \bigwedge_{i_1, \ldots , i_k \in [1, n]} \varphi [\forall a \in \AP : a_{\pi_1} \mapsto a_{i_1}, \ldots , a_{\pi_k} \mapsto a_{i_k}]
\end{align*}
The formula $\varphi^n_\text{LTL}$ enumerates all possible combinations to choose $k$ traces -- one for each quantifier -- from the set of $n$ traces in the model.
We use the notation $\varphi [\forall a \in \AP : a_{\pi_1} \mapsto a_{i_1}, \ldots , a_{\pi_k} \mapsto a_{i_k}]$ to indicate that in $\varphi$, atomic propositions with trace variables are replaced by atomic propositions indexed with one of the $n$ traces. 
We define $I' = \set{a_i ~|~ a \in I, 1 \leq i \leq n}$ and $O'$ analogously. 
Since $n$ is known upfront, we only write $\varphi_\text{LTL}$.

Our algorithm exploits that for every LTL formula $\varphi$, there exists an equivalent parity game $\game_\varphi$ such that $\varphi$ is realizable iff player $P_0$ is winning in the initial state with strategy $\sigma_0$~\cite{DBLP:conf/tacas/EsparzaKRS17}. 
For a finite trace~$u$, $\varphi$~is realizable with prefix~$u$ iff the play induced by~$u$ ends in a state~$q$ that is in the winning region of player~$P_0$.
The algorithm to enforce the HyperLTL formula calls the following three procedures -- depicted in Algorithm~\ref{ParallelInfinite} -- in the appropriate order. 

\begin{algorithm}[t]
	\caption{HyperLTL enforcement algorithm for the parallel input model.}
	\label{ParallelInfinite}
	\begin{multicols}{2}
		\begin{algorithmic}[1]
			\Procedure{Initialize}{$\psi, n$}
			\State $\varphi_\text{LTL}$ := \Call{toLTL}{$\psi, n$};
			\State (game, $q_0$) := \Call{toParity}{$\varphi_\text{LTL}$}; 
			\State winR := \Call{SolveParity}{game};
			\If{$q_0 \notin$ winR}
			\State \textbf{raise} error;
			\EndIf
			\State \textbf{return}(game, winR, $q_0$);
			\EndProcedure	
			\Statex
			
			\Procedure{Enforce}{game, lastq}
			\State sig := \Call{GetStrat}{game, lastq}; 
			\While{true}
			\State o := sig(lastq);
			\State lastq := \Call{Move}{game, lastq, o};
			\State \textbf{output}(o);
			\State i := \Call{getNextInput}{ };
			\State i$_\text{LTL}$ := \Call{toLTL}{i};
			\State lastq := \mbox{\Call{Move}{game, lastq, i$_\text{LTL}$}};
			\EndWhile
			\EndProcedure
			
			\columnbreak
			
			\Procedure{Monitor}{game, winR, q} 
			\State lastq := q;
			\While{true}
			\State o := \Call{getNextOutput}{ };
			\State o$_\text{LTL}$ := \Call{toLTL}{o};
			\State q := \Call{Move}{game, lastq, o$_\text{LTL}$};
			\If{q $\notin$ winR}
			\State \textbf{return}(game, lastq);
			\EndIf
			\State i := \Call{getNextInput}{ };
			\State i$_\text{LTL}$ := \Call{toLTL}{i};
			\State q := \Call{Move}{game, q, i$_\text{LTL}$};
			\State lastq := q;
			\EndWhile	
			\EndProcedure	
		\end{algorithmic}
	\end{multicols}
	\vspace{-10pt}
\end{algorithm}

%\begin{itemize}
	\emph{Initialize:} Construct $\varphi_\text{LTL}$ and the induced parity game $\game_\varphi$. 
	Solve the game $\game_\varphi$, i.e. compute the winning region for player $P_0$. 
	If the initial state $q_0 \in V_0$ is losing, raise an {error}. 
	Otherwise start monitoring in the initial state.
	
	\emph{Monitor:} Assume the game is currently in state $q \in V_0$. 
	Get the next outputs $(o_1, \ldots, o_n) \in O^n$ produced by the $n$ traces of the system and translate them to $o_\text{LTL} \subseteq O'$ by subscripting them as described for formula $\varphi_\text{LTL}$. 
	Move with $o_\text{LTL}$ to the next state. 
	This state is in $V_1$. 
	Check if the reached state is still in the winning region. 
	If not, it is a losing state, so we do not approve the system's output but let the enforcer take over and call {\scshape Enforce} on the last state. 
	If the state is still in the winning region, we process the next inputs $(i_1, \ldots , i_n)$, translate them to $i_\text{LTL}$, and move with $i_\text{LTL}$ to the next state in the game, again in $V_0$. 
	While the game does not leave the winning region, the property is still realizable and the enforcer does not need to intervene. 
	
	\emph{Enforce:} By construction, we start with a state $q \in V_0$ that is in the winning region, i.e., there is a positional winning strategy $ \sigma : V_0 \rightarrow \pow{O'}$ for player $P_0$. 
	Using this strategy,~we output $\sigma (q)$ and continue with the next incoming input $i_\text{LTL}$ to the next state in $V_0$. 
	Continue with this strategy for any incoming input.

%\end{itemize}
\emph{Correctness and Complexity.} By construction, since we never leave the winning region, the enforced system fulfills the specification and the enforcer is sound. 
It is also transparent: 
As long as the prefix produced by the system is not losing, the enforcer does not intervene. 
The algorithm has triple exponential complexity in the number of traces $n$: 
The size of~$\varphi_\text{LTL}$ is exponential in $n$ and constructing the parity game is doubly exponential in the size of~$\varphi_\text{LTL}$~\cite{DBLP:conf/tacas/EsparzaKRS17}. 
Solving the parity game only requires quasi-polynomial time~\cite{DBLP:conf/mfcs/Parys19}.
Note, however, that all of the above steps are part of the initialization.
At runtime, the algorithm only follows the game arena.
If the enforcer is only supposed to correct a single output and afterwards hand back control to the system, the algorithm could be easily adapted accordingly.

%\subsection{Algorithms for the Sequential Trace Input Model}
\subsection{Sequential Trace Input Model}
\label{sec:sequential-algorithms}
Deciding whether a prefix is losing in the sequential model is harder than in the parallel model.
In the sequential model, strategies are defined w.r.t. the traces seen so far -- they incrementally upgrade their knowledge with every new trace.
In general, the question whether there exists a sound and transparent enforcement mechanism for universal HyperLTL specifications is undecidable. 
%
%\vspace{-0ex}
\begin{theorem}
	\label{thm:pcp_reduction}
	In the sequential model, it is undecidable whether a HyperLTL formula $\varphi$ from the universal fragment is enforceable.
\end{theorem}
\begin{proof}
	We encode the classic realizability problem of universal HyperLTL, which is undecidable~\cite{DBLP:conf/cav/FinkbeinerHLST18}, into the sequential model enforcement problem for universal HyperLTL.
	HyperLTL realizability asks if there exists a strategy $\sigma \colon \strat{I}{O}$ such that the set of traces constructed from every possible input sequence satisfies the formula $\varphi$, i.e. whether $\set{(w[0] \cup \sigma(\epsilon)) \cdot (w[1] \cup \sigma(w[0])) \cdot (w[2] \cup \sigma(w[0..1])) \cdot \ldots ~|~ w \in (\pow{I})^\omega}, \emptyset \models \varphi$.	
	Let a universal HyperLTL formula $\varphi$ over $\Sigma = I \cupdot O$ be given. 
	We construct
	$
	\psi \coloneqq ~\varphi~ \land ~
	\forall \pi \ldot \forall \pi' \ldot (\bigwedge_{o \in O} o_\pi \leftrightarrow o_{\pi'}) ~\LTLweakuntil~ (\bigvee_{i \in I} i_\pi \not\leftrightarrow i_{\pi'}) .
	$
	The universal HyperLTL formula $\psi$ requires the strategy to choose the same outputs as long as the inputs are the same. The choice of the strategy must therefore be independent of earlier sessions, i.e., $\sigma(U,\epsilon)(s_I) = \sigma(U', \epsilon)(s_I)$ for all sets of traces $U, U'$ and input sequences $s_I$.
	Any trace set that fulfills $\psi$ can therefore be arranged in a traditional HyperLTL strategy tree branching on the inputs and labeling the nodes with the outputs. 
	Assume the enforcer has to take over control after the first event when enforcing $\psi$.
%	Deciding whether the empty prefix is losing corresponds to finding a suitable prefixed strategy with the empty prefix (see Definition~\ref{def:losing_prefix_unbounded_sequential_model}), i.e., solving the classic reactive synthesis problem.
	Thus, there is a sound and transparent enforcement mechanism for $\psi$ iff $\varphi$ is realizable.
\qed
\end{proof}

\paragraph{Finishing the Current Session.}
%\label{sec:environment_dep_finishing_session}
As the general problem is undecidable, we study the problem where the enforcer takes over control only for the rest of the current session.
For the next session, the existence of a solution is not guaranteed.
%For the start of the next session, the control is handed back to the system. 
This approach is especially reasonable if we are confident that errors occur only sporadically.
We adapt the algorithm presented for the parallel model.
Let a HyperLTL formula $\forall \pi_1 \ldot  \ldots \forall \pi_k \ldot \varphi$ over $\Sigma = I \cupdot O$ be given, where $\varphi$ is quantifier free. 
As for the parallel model, we translate the formula into an LTL formula $\varphi^n_\text{LTL}$. 
We first do so for the first session with $n = 1$. 
We construct and solve the parity game for that formula, and use it to monitor the incoming events and to enforce the rest of the session if necessary. 
For the next session, we construct $\varphi^n_\text{LTL}$ for $n = 2$ and add an additional conjunct encoding the observed trace~$t_1$. 
%joFor example, if $a$ holds on position~$3$ of trace $t_1$, we add the conjunct $\X\X\X a_1$. 
The resulting formula induces a parity game that monitors and enforces the second trace. 
Like this, we can always enforce the current trace in relation to all traces seen so far. 
\begin{algorithm}[t]
	\caption{HyperLTL enforcement algorithm for the sequential trace input model.}\label{EnforceSession}
	\label{alg:sequential}
	\begin{algorithmic}[1]
		\Procedure{EnforceSequential}{$\psi$}
		\State n := 1; 
		\State $\varphi_\text{traces}$ := true; 
		\While{true}
		\State $\psi_\text{curr}$ := \Call{toLTL}{$\psi$, n} $\land ~ \varphi_\text{traces}$;
		\State (game, winR, $q_0$) := \Call{Initialize'}{$\psi_\text{curr}$};
		\State res := \Call{Monitor'}{game, winR, $q_0$};
		\If{res == (`ok', t)}
		\State $\varphi_\text{traces}$ := $\varphi_\text{traces}$ $\land$ \Call{toLTL}{t}; 
		\ElsIf{res = (`losing', t, (game, lastq))} 
		\State t' := \Call{Enforce'}{game, lastq};
		\State $\varphi_\text{traces}$ := $\varphi_\text{traces}$ $\land$ \Call{toLTL}{t $\cdot$ t'}; 
		\EndIf
		\State n++; 
		\EndWhile
		\EndProcedure
	\end{algorithmic}
\end{algorithm}
%
%\vspace{-2ex}
Algorithm~\ref{alg:sequential} depicts the algorithm calling similar procedures as in Algorithm~\ref{ParallelInfinite} (for which we therefore do not give any pseudo code). 
{\scshape Initialize'} is already given an LTL formula and, therefore, does not translate its input to LTL. 
{\scshape Monitor'} returns a tuple including the reason for its termination (`ok' when the trace finished and `losing' when a losing prefix was detected). 
Additionally, the monitor returns the trace seen so far (not including the event that led to a losing prefix), which will be added to $\varphi_\text{traces}$. 
{\scshape Enforce'} enforces the rest of the session and afterwards returns the produced trace, which is then encoded in the LTL formula ({\scshape toLTL}({t})). 

\emph{Correctness and Complexity.}
Soundness and transparency follow from the fact that for the $n$-th session, the algorithm reduces the problem to the parallel setting with $n$ traces, with the first $n-1$ traces being fixed and encoded into the LTL formula $\varphi^n_\text{LTL}$.
We construct a new parity game from $\varphi^n_\text{LTL}$ after each finished session. 
The algorithm is thus of non-elementary complexity. 

\vspace{-1ex}
\paragraph{Safety Specifications.}
If we restrict ourselves to formulas $\psi = \forall \pi_1 \ldots \forall \pi_k \ldot \varphi$, where $\varphi$ is a \emph{safety} formula, we can improve the complexity of the algorithm.
Note, however, that not every property with losing prefixes is a safety property: for the formula $\forall \pi \ldot \LTLglobally (o_\pi \rightarrow \LTLeventually i_\pi)$ with $o \in O$ and $i \in I$, any prefix with $o$ set at some point is losing. However, the formula does not belong to the safety fragment.
%The main idea for an improved algorithm is the following: 
Given a safety formula $\varphi$, we can translate it to a safety game~\cite{DBLP:conf/cav/KupfermanV99} instead of a parity game. 
The LTL formula we create with every new trace is built incrementally, i.e., with every finished trace we only ever add new conjuncts.
With safety games, we can thus recycle the winning region from the game of the previous trace. 
The algorithm proceeds as follows. 
%\begin{itemize}
	1) Translate $\varphi$ into an LTL formula $\varphi^n_\text{LTL}$ for $n=1$. 
	2) Build the safety game $\game^1_{\varphi_\text{LTL}}$ for $\varphi^1_\text{LTL}$ and solve it.  
	Monitor the incoming events of trace $t_1$ as before. 
	Enforce the rest of the trace if necessary.
	3) Once the session is terminated, generate the LTL formula $\varphi^2_\text{LTL} = \varphi^1_\text{LTL} \land \varphi^2_\text{diff}$. 
	As $\varphi^2_\text{LTL}$ is a conjunction of the old formula and a new conjunct $\varphi^2_\text{diff}$, we only need to generate the safety game $\game^2_\text{diff}$ and then build the product of $\game^2_\text{diff}$ with the winning region of $\game^1_{\varphi_\text{LTL}}$. 
	We solve the resulting game and monitor (and potentially enforce) as before. 
%\end{itemize}
The algorithm incrementally refines the safety game and enforces the rest of a session if needed. 
The construction recycles parts of the game computed for the previous session.
We thus avoid the costly translation to a parity game for every new session. 
While constructing the safety game from the LTL specification has still doubly exponential complexity~\cite{DBLP:conf/cav/KupfermanV99}, solving safety games can be done in linear time~\cite{DBLP:journals/tods/Beeri80}.

%\vspace{-2ex}
%===========================================
\section{Experimental Evaluation}
\label{sec:experiments}
%===========================================

We implemented the algorithm for the parallel trace input model in our prototype tool REHyper\footnote{REHyper is available at \url{https://github.com/reactive-systems/REHyper}}, which is written in Rust.
%Our spot-based~\cite{DBLP:conf/atva/Duret-LutzLFMRX16} SAT-solver implementation returns satisfying witnesses.
We use Strix~\cite{meyer2018strix} for the generation of the parity game. 
We determine the winning region and the positional strategies of the game with PGSolver~\cite{friedmann2009pgsolver}.
%For non-reactive systems, i.e., when not distinguishing between inputs and outputs, we monitor the system using RVHyper~\cite{RVHyper}, a runtime verification tool for HyperLTL specifications. 
%We implemented a HyperLTL SAT-solver following the algorithms described in~\cite{EAHyper,deciding_hyperproperties} based on the LTL satisfiability solver Leviathan~\cite{DBLP:conf/ijcai/BertelloGMR16}.
All experiments ran on an Intel Xeon CPU E3-1240~v5 $3.50$~GHz, with $8$ GB memory running Debian $10.6$.
%It processes events sequentially in the RVHyper format~\cite{RVHyper}, i.e., sessions can be arbitrarily started and ended.
%REHyper returns a satisfying continuation of the current session if RVHyper detects a \emph{minimal} bad prefix.
We evaluate our prototype with two experiments. In the first, we enforce a non-trivial formulation of fairness in a mutual exclusion protocol. In the second, we enforce the information flow policy \emph{observational determinism} on randomly generated traces.

%\vspace{-2ex}
\subsection{Enforcing Symmetry in Mutual Exclusion Algorithms}
\begin{table}[t]
	\vspace{-5ex}
	\caption{Enforcing symmetry in the Bakery protocol on pairs of traces. Times are given in seconds.}
	\label{table:results2}
	\def\arraystretch{1.1}
	\setlength{\tabcolsep}{3pt}
	\centering
	\begin{tabular}{r||c|c|c|c||c|c|c|c}
		\hline \hline
		& \multicolumn{4}{c||}{random traces}                              & \multicolumn{4}{c}{symmetric traces}                            \\ \hline \hline
		\multicolumn{1}{r||}{~~$\mid$t$\mid$~~} & \multicolumn{1}{l|}{avg} & \multicolumn{1}{l|}{min} & \multicolumn{1}{l|}{max} & \multicolumn{1}{l||}{\#enforced} & \multicolumn{1}{l|}{avg} & \multicolumn{1}{l|}{min} & \multicolumn{1}{l|}{max} & \multicolumn{1}{l}{\#enforced} \\ \hline
		500                                    & 0.003                              & 0.003                                    & 0.003                              & 0   & 0.013                              & 0.008                                    & 0.020                              & 10                                \\ 
		1000                                   & 0.005                              & 0.005                                   & 0.005                              & 0   & 0.024                              & 0.015                                    & 0.039                              & 10                                \\
		5.000                                   & 0.026                              & 0.024                                   & 0.045                              & 0    & 0.078                              & 0.065                                    & 0.097                              & 10                               \\
		10.000                                   & 0.049                              & 0.047                                   & 0.064                              & 0     & 0.153                              & 0.129                                    & 0.178                              & 10                              \\
	\end{tabular}
\end{table}
Mutual exclusion algorithms like Lamport's bakery protocol ensure that multiple threads can safely access a shared resource. 
To ensure fair access to the resource, we want the protocol to be symmetric, i.e., for any two traces where the roles of the two processes are swapped, the grants are swapped accordingly.
Since symmetry requires the comparison of two traces, it is a hyperproperty.

For our experiment, we used a Verilog implementation of the Bakery protocol~\cite{DBLP:journals/cacm/Lamport74a}, which has been proven to violate the following symmetry formulation~\cite{MCHyper}:
%We used their symmetry specification~\cite{} with a tie-break indicating which process may move first as the following HyperLTL formula.
\begin{align*}
	\forall \pi \ldot \forall \pi' \ldot & (\mathit{pc}(0)_\pi = \mathit{pc}(1)_{\pi'} \land \mathit{pc}(1)_{\pi} = \mathit{pc}(0)_{\pi'}) \LTLweakuntil \hspace{2pt} \neg \hspace{2pt} (\mathit{pause}_\pi = \mathit{pause}_{\pi'} \hspace{2pt} \land\\
	& \quad \mathtt{sym}(\mathit{sel}_\pi, \mathit{sel}_{\pi'}) \land  \mathtt{sym}(\mathit{break}_\pi, \mathit{break}_{\pi'}) \land \mathit{sel}_\pi < 3 \land  \mathit{sel}_{\pi'} < 3) \enspace .
\end{align*}
The specification states that for any two traces, the program counters need to be symmetrical in the two processes as long as the processes are scheduled ($\mathit{select}$) and ties are broken ($\mathit{break}$) symmetrically. 
Both $\mathit{pause}$ and $\mathit{sel} < 3$ handle further implementation details. 
%The specification is not satisfied by the implementation, because the comparison of process IDs is not an atomic operation, which could thus lead to an unfair preference of processes.
The AIGER translation~\cite{MCHyper} of the protocol has 5 inputs and 46 outputs. 
To enforce the above formula, only 10 of the outputs are relevant.
We enforced symmetry of the bakery protocol on simulated pairs of traces produced by the protocol. 
%To that aim, we simulated pairs of traces produced by the AIGER translation of the Verilog implementation.
Table~\ref{table:results2} shows our results for different trace lengths and trace generation techniques. 
We report the average runtime over 10 runs as well as minimal and maximal times along with the number of times the enforcer needed to intervene.
%, we report runtimes of various instances of the setting. 
%Each instance was run 10 times.
%We varied the length of the traces and report 
The symmetry assumptions are fairly specific and are unlikely to be reproduced by random input simulation. 
In a second experiment, we therefore generated pairs of symmetric traces. 
Here, the enforcer had to intervene every time, which produced only a small overhead.

The required game was constructed and solved in $313$ seconds. 
For sets of more than two traces, the construction of the parity game did not return within two hours.
The case study shows that the tool performs without significant overhead at runtime and can easily handle very long traces. 
The bottleneck is the initial parity game construction and solving.
%Furthermore, the tools performs slightly better when enforcing more often: monitoring outputs requires the evaluation of a boolean formula while during enforcement a winning successor state can be chosen arbitrarily.

%\vspace{-3ex}
\subsection{Enforcing Observational Determinism}
\begin{table}[t]
	\vspace{-5ex}
	\caption{Enforcing observational determinism. Times are given in seconds.}
	\label{table:results}
	\def\arraystretch{1.1}
	\centering
	\begin{tabular}{c|c|c|c||c||c|c|c|c||c|c|c|c}
		\hline \hline
		\multicolumn{4}{c||}{benchmark size}                                       &
		\multicolumn{1}{c||}{init time}                                       &
		\multicolumn{4}{c||}{0.5\% bit flip probability}                              & \multicolumn{4}{c}{1\% bit flip probability}                            \\ \hline \hline
		\multicolumn{1}{l|}{\# i}&\multicolumn{1}{l|}{\# o}&\multicolumn{1}{l|}{\# t} &  \multicolumn{1}{c||}{$\mid$t$\mid$} & & \multicolumn{1}{l|}{avg} &  \multicolumn{1}{l|}{min} & \multicolumn{1}{l|}{max} & \multicolumn{1}{l||}{\#enf'ed} & \multicolumn{1}{l|}{avg} & \multicolumn{1}{l|}{min} & \multicolumn{1}{l|}{max} & \multicolumn{1}{l}{\#enf'ed} \\ 
		%		\multicolumn{12}{c}{reactive parallel setting} \\ 
		\multirow{2}{*}{1} & \multirow{2}{*}{1} & \multirow{1}{*}{3}
		& 10000 & 0.517 & 0.014 & 0.013 & 0.017 & 60 & 0.013 & 0.013 & 0.015 & 60 \\
		& & \multirow{1}{*}{8} & 10000  & 65.67 & 0.524 & 0.517 & 0.625 & 99 & 0.524 & 0.517 & 0.588 & 97  \\\hline
		\multirow{2}{*}{2} & \multirow{2}{*}{2} & \multirow{1}{*}{4}
		& 10000 & 0.869 & 0.025 & 0.024 & 0.030 & 73 & 0.032 & 0.031 & 0.043 & 77 \\
		& & \multirow{1}{*}{5}
		& 10000 & 21.189 & 0.038 & 0.037 & 0.041 & 90 & 0.038 & 0.037 & 0.042 & 86 \\ \hline 
		\multirow{3}{*}{3} & \multirow{3}{*}{3} & \multirow{1}{*}{2}
		& 10000 & 0.633 & 0.019 & 0.018 & 0.022 & 47 & 0.023 & 0.023 & 0.025 & 54 \\
		& & \multirow{1}{*}{4} & 5000  & \multirow{2}{*}{132.849} & 0.022 & 0.021 & 0.026 & 77 & 0.021 & 0.021 & 0.021 & 71 \\
		& & \multirow{1}{*}{4}  & 10000 & & 0.038 & 0.036 & 0.056 & 77 & 0.037 & 0.037 & 0.042 & 77 \\ \hline 
		\multirow{3}{*}{4} & \multirow{3}{*}{4} & \multirow{1}{*}{3}
		& 1000 & \multirow{3}{*}{43.885} & 0.010 & 0.008 & 0.015 & 71 & 0.009 & 0.008 & 0.018 & 68 \\
		& & \multirow{1}{*}{3}  & 5000 & & 0.023 & 0.021 & 0.033 & 72 & 0.022 & 0.021 & 0.025 & 78 \\
		& & \multirow{1}{*}{3}  & 10000 & & 0.038 & 0.037 & 0.050 & 76 & 0.038 & 0.037 & 0.041 & 75 \\ 
	\end{tabular}
\end{table}
In our second experiment, we enforced observational determinism, given as the HyperLTL formula 
$\forall \pi \ldot \forall \pi' \ldot (\mathit{o}_\pi \leftrightarrow \mathit{o}_{\pi'}) \LTLweakuntil (\mathit{i}_\pi \nleftrightarrow \mathit{i}_{\pi'})$.
The formula states that for any two execution traces, the observable outputs have to agree as long as the observable inputs agree.
Observational determinism is a prototypical information-flow policy used in many experiments and case studies for HyperLTL (e.g.~\cite{DBLP:conf/cav/FinkbeinerHLST18,MCHyper,DBLP:conf/csfw/BonakdarpourF18}). 
%We enforced observational determinism in the reactive parallel setting (Section~\ref{sec:parallel-algorithms}) and in the non-reactive sequential setting (Appendix~\ref{app:sequential-non-reactive}) where the enforcer may manipulate both inputs and outputs. 
%In the latter setting, the enforcer would be allowed to replace arbitrary inputs once a violation of the specification is observed. 
%In this setting, it is therefore a valid strategy to always generate the same dummy trace. 
%However, transparency prescribes that the system stays in control as long as possible, requiring the enforcer to monitor a growing number of traces.
%However, the non-reactive setting still constitutes an interesting experiment: 
%a system should be in full control as long as possible before an enforcer intervenes. 
%The enforcement mechanism should thus not trivially provides dummy traces from the get-go. 
%On the contrary, the piling number of traces requires a significant monitoring effort.
We generated traces using the following scalable generation scheme: 
At each position, each input and output bit is flipped with a certain probability ($0.5\%$ or $1\%$). 
This results in instances where observational determinism randomly breaks.
Table~\ref{table:results} shows our results.
Each line corresponds to 100 randomly generated instances of the given size (number of inputs/outputs and traces, and length of the sessions).
We report the initialization time that is needed to generate and solve the parity game. 
Furthermore, we report average, minimal, and maximal enforcement time as well as the number of instances where the enforcer intervened. 
All times are reported in seconds. 
%To determine how our implementation scales, we ran the experiments for two different bit flip probabilities of 0.5\% and 1\%, such that our implementation had to intervene frequently.
%We report experiments for the reactive parallel setting and the non-reactive sequential setting.
The bottleneck is again the time needed to construct and solve the parity game. 
At runtime, which is the crucial aspect, the enforcer performs efficiently. 
The higher bit flip probability did not lead to more enforcements: 
For traces of length $10000$, the probability that the enforcer intervenes is relatively high already at a bit flip probability of $0.5\%$.
%

%===========================================
\section{Conclusion}
\label{sec:conclusion}
%===========================================
We studied the runtime enforcement problem for hyperproperties.
%We showed that the trace input model, i.e., how traces are forwarded from the observed system to the enforcement mechanism, has a significant impact on the problem definition.
Depending on the trace input model, we showed that the enforcement problem boils down to detecting losing prefixes and solving a custom synthesis problem.
%To illustrate the relevance of hyperproperty enforcement based on these models, we described practical application scenarios, such as privacy in fitness trackers and fairness in contract signing protocols.
For both input models, we provided enforcement algorithms for specifications given in the universally quantified fragment of the temporal hyperlogic HyperLTL.
While the problem for the sequential trace input model is in general undecidable, we showed that enforcing HyperLTL specifications becomes decidable under the reasonable restriction to only finish the current session.
For the parallel model, we provided an enforcement mechanism based on parity game solving.
Our prototype tool implements this algorithm for the parallel model.
We conducted experiments on two case studies enforcing complex HyperLTL specifications for reactive systems with the parallel model. Our results show that once the initial parity game solving succeeds, our approach has only little overhead at runtime and scales to long traces.

 \bibliographystyle{splncs04}
 \bibliography{main}

\end{document}